\documentclass{article}

\usepackage{graphicx} 
\usepackage{psfrag}
\usepackage{amsmath,amssymb}
\usepackage{amsthm}
\usepackage{verbatim}
\usepackage[usenames,dvipsnames]{color}
\usepackage{epstopdf}
\usepackage{url}
\usepackage{hyperref}

\newtheorem{theorem}{\bf Theorem}

\newcommand{\bqn}{\begin{eqnarray}}
\newcommand{\eqn}{\end{eqnarray}}
\newcommand{\bq}{\begin{eqnarray*}}
\newcommand{\eq}{\end{eqnarray*}}

\begin{document}



\title{
Persistent Homology in Sparse Regression\\
 and Its Application to Brain Morphometry}
\author{Moo K.~Chung
\thanks{{\em Asterisk indicates corresponding author.} M.K. Chung  is with Department of Biostatistics and Medical Informatics and  Waisman Laboratory for Brain Imaging and Behavior, University of Wisconsin, Madison, WI 53706 USA (e-mail: mkchung@wisc.edu).  
Jamie L. Hanson is with the Laboratory of Neurogenetics, Duke University, Durham, NC 27710 USA.
Jieping Ye is with Computer Science and Engineering, Arizona State University.Tempe, AZ 85287 USA.
Richard J. Davidson and Seth D. Pollak are with Waisman Center, University of Wisconsin, Madison, WI 53705 USA. 
}, Jamie L. Hanson,
Jieping Ye,\\  
Richard J. Davidson, Seth D. Pollak
}

\maketitle

%

\begin{abstract}
Sparse systems are usually parameterized by a tuning parameter that determines the sparsity of the system. How to choose the right tuning parameter is a fundamental and difficult problem in learning the sparse system. In this paper, by treating the the tuning parameter as an additional dimension,  persistent homological structures over the parameter space is introduced and explored. The  structures are then further exploited in speeding up the computation using the proposed soft-thresholding technique. 
The topological structures are further used as multivariate features in the tensor-based morphometry (TBM) in characterizing white matter alterations in children who have experienced  severe early life stress and maltreatment. These analyses reveal that stress-exposed children exhibit more diffuse anatomical organization across the whole white matter region.

\end{abstract}


\section{Introduction}

In the usual tensor-based morphometry (TBM), the spatial derivatives of deformation fields obtained during nonlinear image registration for warping individual magnetic resonance imaging (MRI) data to a template is used in quantifying neuroanatomical shape variations \cite{ashburner.2000,thompson.2000,chung.2001.ni}. The Jacobian determinant of a deformation field is most frequently used 
in quantifying the brain tissue growth or atrophy at a voxel level. \cite{davatzikos.1996, machado.1998, dubb.2003} used the Jacobian determinant of the 2D deformation field as a measure of local area-change at each pixel in 2D cross-sections of the corpus callosum.  \cite{thompson.1999,chung.2001.ni} applied the Jacobian of 3D deformations as a measure of the regional growth. 
Subsequently, the statistical parametric maps are obtained by fitting the tensor maps as a response variable in a linear model at each voxel, which results in a massive number of univariate test statistics.

Recently, there have been attempts at explicitly modeling the structural variation of one region to another \cite{rao.2008,cao.1999.correlation,worsley.2005.neural, worsley.2005.royal, lerch.2006, he.2007, he.2008}. This provides additional information that complement existing univariate approaches. In most of these multivariate approaches, anatomical measurements such as mesh coordinates, cortical thickness or Jacobian determinant across different voxels are correlated using models such as canonical correlations \cite{avants.2010,rao.2008}, cross-correlations \cite{cao.1999.correlation,worsley.2005.neural, worsley.2005.royal, lerch.2006, he.2007, he.2008}, partial correlations, which are equivalent to the inverse of covariances \cite{bickel.2008,banerjee.2008,friedman.2008,huang.2009,lee.2011.tmi}. However, these multivariate techniques suffer the small-n large-p problem \cite{friston.1995.MGLM,schafer.2005,valdes.2005,lee.2011.tmi, chung.2013.MICCAI}. Specifically, when the number of voxels are substantially larger than the number of images, it produces an under-determined linear model. The estimated covariance matrix is rank deficient and no longer positive definite. In turn, the resulting correlation matrix is not considered as good approximations to the true correlation matrix.

The small-$n$ large-$p$ problem can be remedied by using sparse methods, which regularize the under-determined linear model with additional sparse penalties. There exist various sparse models: sparse correlation \cite{lee.2011.tmi,chung.2013.MICCAI}, sparse partial correlation \cite{bickel.2008,huang.2009,lee.2011.tmi}, sparse canonical correlation \cite{avants.2010} and L1-norm penalized log-likelihood \cite{banerjee.2006,banerjee.2008,friedman.2008, huang.2010,mazumder.2012,witten.2011}. Sparse model $\mathcal{A}(\lambda)$ is usually parameterized by a tuning parameter $\lambda$ that controls the sparsity of the representation. Increasing the sparse parameter makes the  solution more sparse. So far, all previous sparse network approaches use a fixed parameter $\lambda$ that may not be optimal. 
Depending on the choice of the sparse parameter, the final statistical results will be different. Instead of performing statistical inference at one fixed sparse parameter $\lambda$ that may not be optimal,  we introduce a new framework that performs statistical inferences over the whole parameter space using  persistent homology \cite{carlsson.2008,edelsbrunner.2002,singh.2008,ghrist.2008,chung.2009.MICCAI, 
chung.2013.MICCAI, lee.2011.MICCAI, lee.2012.tmi}. 

Persistent homology is a recently popular branch of computational topology with
applications in protein structures \cite{sacan.2007}, gene
expression \cite{dequeant.2008}, brain cortical thickness \cite{chung.2009.MICCAI}, activity patterns in visual cortex \cite{singh.2008}, sensor networks \cite{deSilva.2007}, complex networks \cite{horak.2009} and brain networks \cite{lee.2011.MICCAI,lee.2012.tmi}. 
However, as far as we are aware, it is yet to be applied to sparse models in any context. This is the first study that introduce persistent homology in sparse models. The proposed persistent homological framework is similar to the existing multi-thresholding framework that has been used in modeling connectivity matrices at many different thresholds \cite{achard.2006,he.2008,supekar.2008,lee.2012.tmi}. However, such approach has not been applied in sparse networks before. In a sparse network, sparsity is controlled by the sparse parameter $\lambda$ and the estimated sparse matrix entries. So it is unclear how the existing multi-thresholding framework can be applicable in this situation. In this paper, we establish that thresholding the sparse parameter is equivalent to thresholding correlations under some conditions. Thus, we resolve the unclarity of applying the existing multi-threshold method to the sparse networks.

The main methodological contributions of this paper are as follows.
(i) We introduce a new sparse model based on Pearson correlation. 
 Although various sparse models have been proposed for other correlations such as partial correlations \cite{bickel.2008,huang.2009,lee.2011.tmi} and canonical correlations  \cite{avants.2010},  the sparse version of the Pearson correlation was not often studied. 
 
(ii) We introduce persistent homology in the proposed sparse model for the first time. We explicitly show that persistent homological structures can be found in the sparse model. This paper  differs substantially from our previous study \cite{lee.2012.tmi}, which studies the persistent homology in graphs and networks. However, sparse models and sparse networks were never considered previously. 

(iii) We show that the identification of persistent homological structures can yield greater computational speed and efficiency in solving the proposed sparse Pearson correlation model without any numerical optimization. Note that most sparse models require numerical optimization for minimizing L1-norm penalty, which can be a computational bottleneck for solving large scale problems. There are few attempts at speeding up the computation for sparse models. By identifying block diagonal structures in the estimated (inverse) covariance matrix, it is possible to bypass the numerical optimization in the penalized log-likelihood method \cite{mazumder.2012,witten.2011}. LASSO (least absolute shrinkage and selection operator) can be done without numerical optimization if the design matrix is orthogonal \cite{tibshirani.1996}. The proposed method substantially differs from \cite{mazumder.2012,witten.2011} in that we do not need to assume the data to follow normality since there is no need to specify the likelihood function.  Further the cost functions we are optimizing are different. The proposed method also differs from \cite{tibshirani.1996} in that our problem is not orthogonal. 

As an application of the proposed method, we applied the techniques to the quantification of interregional white matter abnormality in stress-exposed children's magnetic resonance images (MRI). Early and severe childhood stress, such as experiences of abuse and neglect, have been associated with a range of cognitive deficits \cite{pollak.2008,sanchez.2009, loman.2013} and structural abnormalities \cite{jackowski.2009, hanson.2012, hanson.2013}. 
However, little is known about the underlying biological mechanisms leading to cognitive problems in these children \cite{pollak.2010} due to the difficulties in the existing methods that do not have enough discriminating power. However, we demonstrate that the proposed method is very well suited to this problem.


\section{Methods}

\subsection{Sparse Correlations}
\label{sec:SC}



{\em Correlations.}
Consider measurement vector ${\bf x}_j$ on node $j$. If we center and rescale the measurement ${\bf x}_j$ such that 
$$\parallel {\bf x}_j \parallel^2 = {\bf x}_{j}' {\bf x}_{j} =1,$$ 
the sample correlation between nodes $i$ and $j$ is given by
${\bf x}_{i}' {\bf x}_{j}$. Since the data is normalized, the sample covariance matrix is reduced to the sample correlation matrix. 

Consider the following linear regression between nodes $j$ and $k$ $(k \neq j)$:
\bqn {\bf x}_j= \gamma_{jk} {\bf x}_k + \epsilon_j. \label{eq:LRG}\eqn
We are basically correlating data at node $j$ to data at node $k$. In this particular case, $\gamma_{jk}$ is the usual Pearson correlation. The least squares estimation (LSE) of $\gamma_{jk}$ is then given by 
\bqn \widehat{\gamma}_{jk} = {\bf x}_j' {\bf x}_k, \label{eq:gamma}\eqn
which is the sample correlation. For the normalized data, regression coefficient estimation is exactly the sample correlation. For the normalized and centered data, the regression coefficient is the correlation. It can be shown that (\ref{eq:gamma}) minimizes the sum of least squares over all nodes:
\bqn \sum_{j=1}^p \sum_{k \neq j}
\parallel {\bf x}_{j}  - \gamma_{jk} {\bf x}_{k} \parallel^{2}. \label{eq:LRG2}\eqn
Note that we do not really care about correlating ${\bf x}_j$ to itself since the correlation is then trivially $\gamma_{jj}=1$. \\

{\em Sparse Correlations.}
Let ${\bf \Gamma} = (\gamma_{jk})$ be the correlation matrix. The sparse penalized version of (\ref{eq:LRG2}) is given by
\bqn 
\label{eq:lasso_corr}
F ({\bf \Gamma})= 
\frac{1}{2}\sum_{j=1}^p \sum_{k \neq j}
\parallel {\bf x}_{j}  - \gamma_{jk} {\bf x}_{k} \parallel^{2} + \lambda \sum_{j=1}^p \sum_{k \neq j}  | \gamma_{jk} |.
\eqn
The sparse correlation is given by minimizing $F ({\bf \Gamma})$. By increasing $\lambda$, the estimated correlation matrix $\widehat{\bf \Gamma}(\lambda)$ becomes more sparse. When $\lambda=0$, the sparse correlation is simply given by the sample correlation, i.e. $\widehat{\gamma}_{jk} = {\bf x}_j' {\bf x}_k$.  As $\lambda$ increases, the correlation matrix ${\bf \Gamma}$ shrinks to zero and becomes more sparse. 
This is separable compressed sensing or LASSO  type problem. 
Further, there is no need to numerically optimize (\ref{eq:lasso_corr})  using the coordinate descent learning or the active-set algorithm often used in compressed sensing  \cite{peng.2009,friedman.2008}. The minimization of (\ref{eq:lasso_corr}) can be done by the proposed soft-thresholding method analytically by exploiting the topological structure of the problem. This sparse regression is not orthogonal, i.e. ${\bf x}_{i}' {\bf x}_{j} \neq \delta_{ij}$, the Dirac delta, so the existing soft-thresholding method for LASSO \cite{tibshirani.1996} is not applicable.\\


\begin{theorem} 
For $\lambda \geq 0$, the solution of the following separable LASSO problem
$$\widehat{\gamma}_{jk}(\lambda) = \arg \min_{\gamma_{jk}}  \frac{1}{2}\sum_{j=1}^p  \sum_{k \neq j} \parallel {\bf x}_{j}  - \gamma_{jk} {\bf x}_{k} \parallel^{2} + \lambda \sum_{j=1}^p  \sum_{k \neq j} | \gamma_{jk} |,$$
is given by the soft-thresholding
 \bqn \widehat{\gamma}_{jk}(\lambda) 
 = \begin{cases} 
 {\bf x}_{j}' {\bf x}_{k}  - \lambda & \mbox{ if }   {\bf x}_{j}' {\bf x}_{k}  > \lambda \\
 0                                               & \mbox{ if }   |{\bf x}_{j}' {\bf x}_{k}|  \le \lambda \\ 
{\bf x}_{j}' {\bf x}_{k}  + \lambda & \mbox{ if }   {\bf x}_{j}' {\bf x}_{k}  < -\lambda
\end{cases}.
\label{eq:lambda-without2}
\eqn
\end{theorem}
\begin{proof}
Write (\ref{eq:lasso_corr}) as
\bqn 
\label{eq:lasso_corr2}
F ({\bf \Gamma})= \frac{1}{2}\sum_{j=1}^p  \sum_{k \neq j} f(\gamma_{jk}),\eqn 
where 
$$f(\gamma_{jk}) = 
\parallel {\bf x}_{j}  - \gamma_{jk} {\bf x}_{k} \parallel^{2} + 2 \lambda  | \gamma_{jk} |.$$
Since $f(\gamma_{jk})$ is nonnegative and convex, $F ({\bf \Gamma})$ is minimum if each component $f(\gamma_{jk})$ achieves minimum. So we only need to minimize each component $f(\gamma_{jk})$.
This differentiates 
our sparse correlation formulation from the standard compressed sensing that cannot be optimized in this component wise fashion. $f(\gamma_{jk})$ can be rewritten as
\bq f(\gamma_{jk}) &=& 
\| {\bf x}_{j}\|^2   - 2 \gamma_{jk} {\bf x}_{j}' {\bf x}_{k} + \gamma_{jk}^2 \| {\bf x}_{k} \|^{2} + 2 \lambda  | \gamma_{jk} | \\
&=&(\gamma_{jk} - {\bf x}_{j}' {\bf x}_{k} )^2 + 2 \lambda  | \gamma_{jk} | +1.
\eq
We used the fact ${\bf x}_{j}' {\bf x}_{j} =1.$ 
  
For $\lambda=0$, the minimum of $f(\gamma_{jk})$ is achieved when
$\gamma_{jk} = {\bf x}_j' {\bf x}_k$, which is the usual LSE.
For $\lambda > 0$, 
Since $f(\gamma_{jk})$ is quadratic in $\gamma_{jk}$, the minimum is achieved when
\bqn 
\frac{\partial f}{\partial \gamma_{jk}} =  
 2 \gamma_{jk} -  2 {\bf x}_{j}' {\bf x}_{k}  \pm 2 \lambda =0 \label{eq:pm}\eqn
The sign of $\lambda$ depends on the sign of $\gamma_{jk}$. Thus, sparse correlation $\widehat{\gamma}_{jk}$ is given by a soft-thresholding of $ {\bf x}_{j}' {\bf x}_{k}$:
 \bqn \widehat{\gamma}_{jk}(\lambda) 
 = \begin{cases} 
 {\bf x}_{j}' {\bf x}_{k}  - \lambda & \mbox{ if }   {\bf x}_{j}' {\bf x}_{k}  > \lambda \\
0                                               & \mbox{ if }   |{\bf x}_{j}' {\bf x}_{k}|  \le \lambda \\ 
{\bf x}_{j}' {\bf x}_{k}  + \lambda & \mbox{ if }   {\bf x}_{j}' {\bf x}_{k}  < -\lambda
\end{cases}.
\label{eq:lambda-without2}
\eqn
\end{proof}


Theorem 1 is heuristically explained in the conference paper \cite{chung.2013.MICCAI} without the rigorous proof or statement. This paper extends \cite{chung.2013.MICCAI} with clearly spelled out soft-thresholding rule and the detailed proof. The  estimated sparse correlation (\ref{eq:lambda-without2}) basically thresholds the sample correlation that is larger or smaller than  $\lambda$ by the amount $\lambda$. Due to this simple expression,  there is no need to optimize (\ref{eq:lasso_corr}) numerically as often done in compressed sensing or LASSO \cite{peng.2009, friedman.2008}. However, Theorem 1 is only  applicable to separable cases and for non-separable cases, numerical optimization is still needed.

Since different choices of sparsity parameter $\lambda$ will produce different solutions in sparse model $\mathcal{A}(\lambda)$, we propose to use the collection of  all the sparse solutions for many different values of  $\lambda$ for the subsequent statistical inference. This avoids the problem of identifying the optimal sparse parameter that may not be optimal in practice. The question is then how to use the collection of $\mathcal{A}(\lambda)$ in a coherent mathematical fashion. For this, we propose 
to apply persistent homology  \cite{edelsbrunner.2008,lee.2011.MICCAI,lee.2012.tmi} and establish Theorem 2.

\subsection{Persistent  Homology in Graphs}
\label{sec:PHG}

Using persistent homology, topological features such as the connected components and cycles of a graph can be tabulated in terms of the Betti numbers. The zeroth Betti number $\beta_0$ and the first Betti number $\beta_1$, which are topological invariants, respectively counts the the number of connected components and holes in the graph \cite{edelsbrunner.2002}. The network difference is then quantified using the Betti numbers of the graph \cite{lee.2011.MICCAI, lee.2012.tmi}.  The graph filtration is a new graph simplification technique that iteratively builds a nested subgraphs of the original graph. The algorithm simplifies a complex graph by piecing together the patches of locally connected nearest nodes. The process of graph filtration is mathematically equivalent to the single linkage hierarchical clustering and dendrogram construction \cite{lee.2011.MICCAI, lee.2012.tmi}.

\begin{figure}[t]
\begin{center}
\includegraphics[width=1\linewidth]{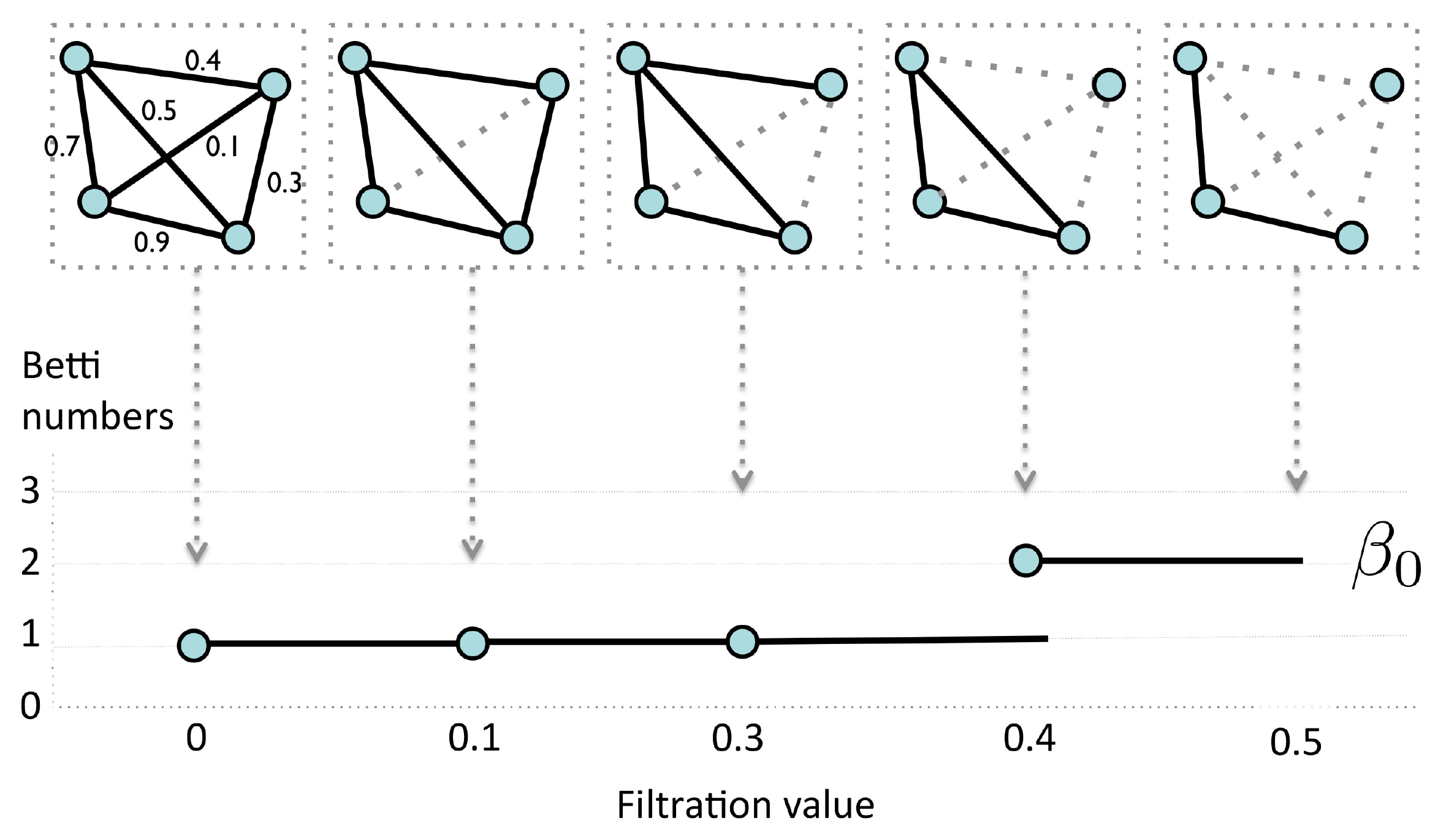}
\caption{Schematic of  graph filtration. We start with a weighted graph (top left). 
We sort the edge weights in an increasing order. We threshold the graph at filtration value $\lambda$ and obtain unweighted binary graph $\mathcal{G}(\lambda)$ based on rule (\ref{eq:binarythreshold}). The thresholding is performed sequentially by increasing $\lambda$ values. Then we obtain the sequence of nested graphs such as $\mathcal{G}(0.0) \supset \mathcal{G}(0.1) \supset \mathcal{G}(0.3) \supset \cdots$. The collection of such nested graph is defined as a graph filtration. The dotted lines are thresholded edges. 
The first Betti number $\beta_0$, which counts the number of connected components, is then plotted over the filtration.}
\label{fig:connectivity-weighted}
\end{center}
\end{figure}

Consider a general setting of a weighted graph with node set $V = \left\{ 1, \dots, p \right\}$ and edge weights $\rho=(\rho_{jk})$, where $\rho_{jk}$ is the weight between nodes $j$ and $k$. Weighted graph $X = (V, \rho)$ is formed by the pair of node set $V$ and edge weights $\rho$. The edge weights in many brain imaging applications are usually given by some similarity measures such as correlation or covariance between nodes \cite{lee.2011.MICCAI,li.2009,mcintosh.1994,newman.1999,song.2005}. 


Given a weighted network $X=(V, \rho)$, we induce binary network $\mathcal{G}(\lambda)$ by thresholding the weighted network at $\lambda$. The adjacency matrix $A=(a_{jk})$ of $\mathcal{G}(\lambda)$ is defined as 
\bqn a_{jk}(\lambda) = 
\begin{cases}
1 &\; \mbox{  if  }  \rho_{jk} > \lambda;\\
0 & \; \mbox{ otherwise.}
\end{cases}
\label{eq:binarythreshold}
\eqn
Any edge weight less than or equal to $\lambda$ is made into zero while edge weight larger than $\lambda$ is made into one. The binary network  $\mathcal{G}(\lambda)$ is a simplicial complex consisting of $0$-simplices (nodes) and $1$-simplices (edges), a special  case of the Rips complex \cite{ghrist.2008}. Then it can be easily seen that $\mathcal{G}(\lambda_1)  \supset \mathcal{G}(\lambda_2)$
for $\lambda_1 < \lambda_2$ in a sense the vertex and edge sets of $\mathcal{G}(\lambda_2)$ are contained in those of $\mathcal{G}(\lambda_1)$. 
Therefore, just as in the case of Rips filtration, which is a collection of nested Rips complexes, we can construct the filtration on the collection of binary networks:
\bqn \mathcal{G}(\lambda_0)  \supset \mathcal{G}(\lambda_1)\supset \mathcal{G}(\lambda_2) \supset \cdots  \label{eq:Gfiltration} \eqn
for $0=\lambda_0 < \lambda_1 < \lambda_2 < \cdots.$ 
Note that $\mathcal{G}(0)$ is the complete weighted graph while $\mathcal{G}(\infty)$ gives the  node set $V$. By increasing the $\lambda$ value, we are thresholding at higher correlation so more edges are removed and thin out the connections. Such the nested sequence of the Rips complexes (\ref{eq:Gfiltration}) is called a Rips filtration, the main object of interest in  persistent homology \cite{edelsbrunner.2008}. The sequence of $\lambda$ values are called the filtration values. Since we are dealing with a special case of Rips complexes restricted to graphs, we will call such filtration {\em graph filtration}.  Figure \ref{fig:connectivity-weighted} illustrates an example of a graph filtration with 4 nodes. Sequentially we are deleting edges based on the ordering of the edge weights. Since the graph filtration is a special case of the Rips filtration, it inherits all the topological properties of the Rips filtration. Given a weighted graph, there are infinitely many different filtrations. In Figure \ref{fig:connectivity-weighted} example, we have two filtrations $\mathcal{G}(0.0) \supset \mathcal{G}(0.1) \supset \mathcal{G}(0.3) \supset  \mathcal{G}(0.4) \supset  \mathcal{G}(0.5) $
 and $\mathcal{G}(0.0) \supset \mathcal{G}(0.2) \supset \mathcal{G}(0.6)$ among many other possiblities. 
So a question naturally arises if there is a unique filtration that can be used in characterizing the graph. Let the {\em level of a filtration} be the number of nested unique sublevel sets in the given filtration.

\begin{theorem} 
\label{theorem:maximal}
For graph $X=(V, \rho)$ with $q$ unique edge weights, the maximum level of a filtration on the graph is $q+1$. Further, the  filtration with $q+1$ filtration level is unique.
\end{theorem}
\begin{proof}
For a graph with $p$ nodes, the maximum number of edges is $(p^2-p)/2$, which is obtained in a complete graph. If we order the edge weights in the increasing order, we have the sorted edge weights:
$$0 = \rho_{(0)} <  \min_{j,k} \rho_{jk} = \rho_{(1)} < \rho_{(2)} < \cdots < \rho_{(q)} = \max_{j,k} \rho_{jk},$$
where $q \leq (p^2-p)/2$.  The subscript $_{( \;)}$ denotes the order statistic. 
For all $\lambda < \rho_{(1)}$, $\mathcal{G}(\lambda) = \mathcal{G}(0)$ is the complete graph of $V$. For all $\rho_{(r)}  \leq \lambda < \rho_{(r+1)} \; (r =1, \cdots, q-1)$, $\mathcal{G}(\lambda) = \mathcal{G}(\rho(r))$. For all $ \rho_{(q)} \leq \lambda$, $\mathcal{G}(\lambda) = \mathcal{G}(\rho_{(q)}) =V$, the vertex set. Hence, the filtration given by
\bqn  \mathcal{G}(0)  \supset  \mathcal{G}(\rho_{(1)})  \supset  \mathcal{G}(\rho_{(2)})  \supset \cdots  \supset  \mathcal{G}(\rho_{(q)})\label{eq:maximal}\eqn
is {\em maximal} in a sense that we cannot have any additional level of filtration. 
\end{proof}

Among many possible filtrations, we will use the maximal filtration (\ref{eq:maximal}) in the study since it is uniquely given. The finiteness and uniqueness of the filtration levels over finite graphs are intutively clear by themselves and are already applied in software packages such as javaPlex. \cite{adams.2014}. However, we still need a rigorous statement to specify the type of filtration we are using out of many.


\subsection{Persistent Homology in Sparse Regression}

\begin{figure}[t]
\centering
\includegraphics[width=1\linewidth]{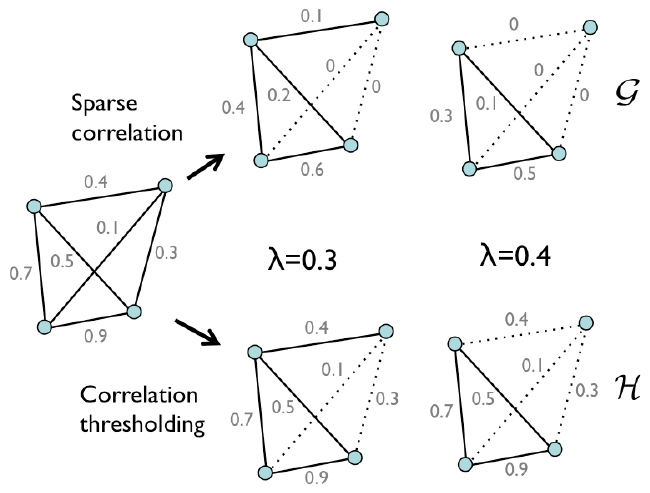}
\caption{Comparison between the sparse correlation estimation via numerical optimization (top) and the proposed soft-thresholding method (bottom). The direct numerical optimization makes the graph sparse by shrinking the edge weights to zero. Nonzero edges form binary graph $\mathcal{G}$. The soft-thresholding method thresholds the sample correlations at given filtration value and construct binary graph $\mathcal{H}$. The both methods produce the identical binary graphs, i.e. $\mathcal{G} = \mathcal{H}$.  
If the methods are applied  at two different parameters $\lambda = 0.3, 0.4$, we obtain nested binary graphs $\mathcal{G}(0.3) \supset \mathcal{G}(0.4)$ and $\mathcal{H}(0.3) \supset \mathcal{H}(0.4)$. Theorem 3 generalizes this example.
}
\label{fig:jacobian-schematiccorr}
\end{figure}

We introduce a persistent homological structure in sparse correlations now as follows. 
Let $A = (a_{jk}(\lambda))$ be the adjacency matrix obtained from sparse correlation (\ref{eq:lambda-without2}):
$$a_{jk}(\lambda) = 
\begin{cases}
1 &\; \mbox{  if  }  \widehat{\gamma}_{jk}(\lambda)  \neq 0;\\
0 & \; \mbox{ otherwise.} 
\end{cases}
$$
Let  $\mathcal{G}(\lambda)$ be the graph defined by the adjacency matrix $A$. Then we have the main result of this paper, which relies on the results of Theorem 1 and Theorem 2.

\begin{theorem}
\label{theorem:PHG}
For centered and normalized data ${\bf x}_j \; (j=1, \cdots, p)$, $\rho_{(1)}, \rho_{(2)}, \cdots, \rho_{(q)}$ be the order statistic of $|{\bf x}_j'{\bf x}_k| \; (1 \leq j,k \leq p, k\neq j)$, i.e. the sorted sequence of $|{\bf x}_j'{\bf x}_k|$ in increasing order. Then graph  $\mathcal{G}(\lambda)$ obtained from the sparse regression (\ref{eq:lasso_corr}) forms the maximal graph filtration
\bqn  \mathcal{G}(0) \supset \mathcal{G}(\rho_{(1)}) \supset \mathcal{G}(\rho_{(2)}) \supset \cdots  \supset  \mathcal{G}(\rho_{(q)}). 
\label{eq:PHG}
\eqn 
\end{theorem}
\begin{proof}
The proof follows by simplifying the adjacency matrix $A$ into a simpler but equivalent adjacency matrix $B=(b_{jk})$. From Theorem 1, $\widehat{\gamma}_{jk} (0) \neq 0$ if $ |{\bf x}_{j}' {\bf x}_{k}| > \lambda$ and 0 otherwise. Thus, the adjacency matrix $A$ is equivalent to the adjacency matrix $B=(b_{jk})$:
\bqn b_{jk}(\lambda) = 
\begin{cases}
1 &\; \mbox{  if } |{\bf x}_j'{\bf x}_k | > \lambda;\\
0 & \; \mbox{ otherwise.}
\end{cases} \label{eq:Bcases}\eqn
Let $\mathcal{H}(\lambda)$ be the graph defined by adjacency matrix $B$. Graph $\mathcal{H}(\lambda)$ is formed by thresholding edge weights given by the absolute value of sample correlations ${\bf x}_j'{\bf x}_k$. From Theorem 2, 
such graph must have maximal filtration:
\bqn \mathcal{H}(0) \supset \mathcal{H}(\rho_{(1)}) \supset \mathcal{H}(\rho_{(2)}) \supset \cdots  \supset  \mathcal{H}(\rho_{(q)}). 
\label{eq:PHG2}
\eqn 
Since $A=B$, graph $\mathcal{G}$ also must have the identical maximal filtration. This proves the statement.
\end{proof}

Theorem \ref{theorem:PHG} is illustrated in  Figure \ref{fig:jacobian-schematiccorr}. In obtaining the topological structure of a graph induced by sparse correlation, it is not necessary to solve the sparse regression by the direct optimization, which can be very time consuming. Identical topological information can be obtained by performing the soft-thresholding on the sample correlations. Figure \ref{fig:jacobian-schematiccorr} illustrates how Theorems \ref{theorem:PHG} is used to construct a sparse correlation network using a 4-nodes example. In the application, $p=548$ nodes will be used.

The resulting  maximal filtration can be quantified by plotting the change of Betti numbers over increasing filtration values \cite{edelsbrunner.2002,ghrist.2008, lee.2011.MICCAI}. The first Betti number $\beta_0(\lambda)$ counts the number of connected components of the given graph $\mathcal{G}(\lambda)$ at the filtration value $\lambda$ \cite{lee.2012.tmi}. Given graph filtration $\mathcal{G}(\lambda_0) \supset  \mathcal{G}(\lambda_1) \supset \mathcal{G}(\lambda_2) \supset \cdots$, we plot the first Betti numbers $\beta_0(\lambda_0) < \beta_0(\lambda_1)  < \beta_0(\lambda_2) \cdots $ over filtration values $\lambda_0 < \lambda_1 < \lambda_2 \cdots$ (Figure \ref{fig:connectivity-weighted}). The number of connected components increase as the filtration value increases. The pattern of increasing number of connected components visually show how the graph structure changes over different parameter values. The overall pattern of  Betti (number) plots can be used as a summary measure of quantifying how the graph changes over increasing edge weights. The Betti number plots are related but different from {\em barcodes} in litearture. The Betti number is equal to the number of bars in the barcodes at the specific filtration value. It is not necessary to perform filtrations for infinitely many possible $\lambda$ values in plotting the Betti numbers. From Theorem \ref{theorem:maximal}, the maximum possible number of filtration level for  plotting the Betti numbers is $q+1$, where $q$ is the number of unique edge weights. For a tree, which is a graph with no cycle, we can come up with a much stronger statement than this.

\begin{theorem}
\label{theorem:barcodes}
For a tree with $p \geq 2$ nodes and unique positive edge weights $\rho_{(1)} < \rho_{(2)} < \cdots < \rho_{(p-1)}$, the plot for the first Betti number ($\beta_0$) corresponding to the maximal graph filtration is given by the coordinates
$$(0, 1), (\rho_{(1)}, 2), \cdots,  (\rho_{(2)}, 3), (\rho_{(p-1)}, p), (\infty, p).$$
\end{theorem}
\begin{proof}
For a tree with $p$ nodes, there are total $p-1$ edges.  Then from Theorem \ref{theorem:maximal}, we have the maximal filtration 
\bqn  \mathcal{G}(\rho_{(0)})  \supset  \mathcal{G}(\rho_{(1)})  \supset  \mathcal{G}(\rho_{(2)})  \supset \cdots  \supset  \mathcal{G}(\rho_{(p-1)}).\label{eq:maximal2}\eqn
Since all the edge weights are above filtration value $\rho_{(0)}=0$, all the nodes are connected, i.e., 
$\beta_0(\rho_{(0)}) = 1$. Since no edge weight is above the threshold  $\rho_{(q-1)}$, $\beta_0( \rho_{(p-1)}) = p$. 
At each time we threshold an edge, the number of components increases exactly by one in the tree. Thus,
we have 
$$\beta_0(\rho_{(1)}) = 2, \beta_0(\rho_{(2)}) = 3, \cdots, \beta_0(\rho_{(p-1)}) = p.$$ 
\end{proof}
For a general graph, it is not possible to analytically determine the coordinates for its Betti-plot. The best we can do is to compute the number of connected components $\beta_0$ numerically using the single linkage dendrogram method (SLD) \cite{lee.2012.tmi}, the Dulmage-Mendelsohn decomposition \cite{pothen.1990, chung.2011.spie} or existing simplical complex approach \cite{deSilva.2007,carlsson.2008,edelsbrunner.2002}. For our study, we used the SLD method.

\begin{figure*}[t]
\centering
\includegraphics[width=1\linewidth]{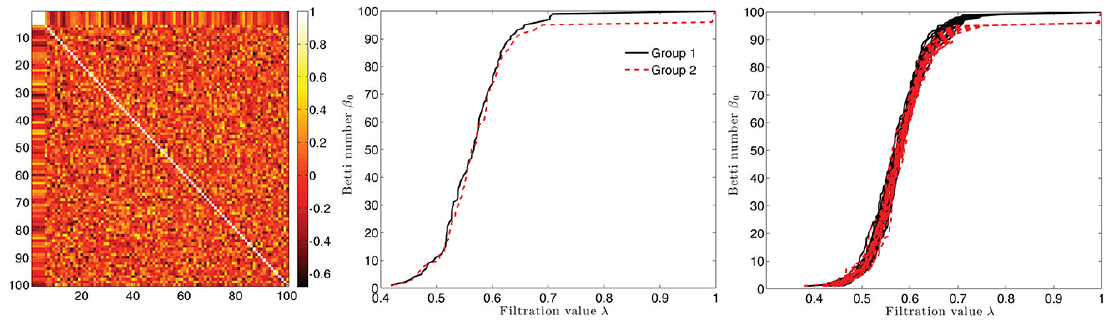}
\caption{Simulation study. Left: the simulated correlation matrix for Group 2, where the first 5 nodes are connected (white square). Group 1 has no connection. Middle: The resulting $\beta_0$-plot showing group difference. Right: Leave-one-out Jackknife resampled $\beta_0$-plots of Group 1(solid line) and Group 2 (dotted line).
Then the rank-sum test is performed on the area differences under $\beta_0$-curves between the groups ($p$-value $<$ 0.001). The statistically significant result corresponds to the horizontal gap in the Betti numbers after filtration value 0.7.} 
\label{fig:jackknife}
\end{figure*}

\subsection{Statistical Inference on Betti number plots}

The first Betti number will  be used as features for characterizing network differences statistically. We assume there are $n$ subjects and $p$ nodes in Group 1. For subject $i$, we have measurement $x_{ij}$ at node $j$. Denote data matrix as $X = (x_{ij})$, where $x_{ij}$ is the measurement for subject $i$ at node $j$. We then construct a sparse network and corresponding Betti number $\beta_{0}^1(\lambda)$ using $X$. Thus,  $\beta_{0}^1(\lambda)$ is a function of $X$. Consider another Group 2 consists of $m$ subjects. For Group 2, data matrix is denoted as $Y = (y_{ij})$, where $y_{ij}$ is the measurement for subject $i$ at node $j$. Group 2 will also generate single Betti number plot $\beta_{0}^2(\lambda)$ as a function of $Y$. We are then interested in testing if the shapes of Betti number plots are different between the groups. This can be done by comparing the areas under the Betti plots. So the null hypothesis of interest is 
\bqn H_0: \int_{0}^{1} \beta_0^1 (\lambda) \;d\lambda = \int_{0}^{1} \beta_0^2 (\lambda) \; d\lambda \label{eq:H0}\eqn
while the alternate hypothesis is
$$H_1: \int_{0}^{1} \beta_0^1 (\lambda) \;d\lambda \neq \int_{0}^{1} \beta_0^2 (\lambda) \; d\lambda.$$
This inference avoids the use of multiple comparisons. The null hypothesis (\ref{eq:H0}) is related to the following pointwise null hypothesis:
\bqn H_0':  \beta_0^1 (\lambda)  =  \beta_0^2 (\lambda) \mbox{ for all } \lambda \in [0,1] \label{eq:H0'} \eqn
If the hypothesis (\ref{eq:H0'}) is true, the hypothesis (\ref{eq:H0}) is also true (but inverse is not true). Thus, testing the area under the curve is related to testing the height of the curve at every point. The advantage of using the area under the curve is that we do not need to worry about multiple comparisons associated with testing (\ref{eq:H0'}).  The area under the curve seems a reasonable approach to use for Betti-plots. A similar approach has been introduced in \cite{chung.2001.thesis} in removing the multiple comparisons and produce a single summary test statistic.

There is no prior study on the statistical distribution on the Betti numbers so it is difficult to construct a parametric test procedure. Further, since there is only one Betti-plot per group, it is not even possible to construct a statistic without resampling techniques.  So it is necessary to empirically construct the null distribution and determine the p-value by  resampling techniques such as the permutation test and jackknife \cite{chung.2013.MICCAI,chung.2013.SCM,efron.1982,lee.2012.tmi}. For this study, we use the jackknife resampling. 

For Group 1 with $n$ subjects, one subject is removed at a time and the remaining $n-1$ subjects are used in constructing a network and a Betti-plot. 
Let $X_{-l}$ be the data matrix, where the $l$-th row (subject) is removed from $X$. 
Then for each $l$-th subject removed, we compute $\beta_0^{1(-l)}$, which is a function of $\lambda$ and $X_{-l}$. Repeating this process for each subject, we obtain $n$  Betti-plots $\beta_0^{1(-1)}, \beta_0^{1(-2)}, \cdots, \beta_0^{1(-n)}$. 
For Group 2, the $l$-th row (subject) is removed from the original data matrix $Y$ and obtain data matrix as $Y_{-l}$. For each $l$-th subject removed, we compute $\beta_0^{2(-l)}$, which is a function of $\lambda$ and $X_{-l}$. Repeating this process for each subject, we obtain $m$ Betti-plots $\beta_0^{2(-1)}, \beta_0^{2(-2)}, \cdots, \beta_0^{2(-m)}$. There are 23 maltreated and 31 control children in the study, so we have 23 and 31 Jackknife resampled Betti-plots. Subsequently we compute the areas under the Betti-plots by discretizing the integral and doing the Riemann sum. The area differences between the groups are then tested using the Wilcoxon rank-sum test, which is a nonparametric test on median differences \cite{gibbons.2011}.

We did not use the permutation test. For the permutation test to converge for our data set, it requires tens of thousands permutations and it is really time consuming even with the proposed time-saving soft-thresholding method. The proposed method takes about a minute of computation in a desktop but ten-thousands permutations will take about seven days of computation. Hence, we used a much simpler Jackknife resampling technique. The procedure is validated using the simulation with the known ground truth.
{\tt MATLAB} codes for constructing network filtration, barcodes and performing statical inference on are given in \url{http://brainimaging.waisman.wisc.edu/~chung/barcodes} with the post-processed Jacobian determinant and FA data that was used for this study.

\noindent 
{\em Simulations.} We performed two simulations. In each simulation, the sample sizes are $n=20$ in Group 1 and $m=20$ in Group 2. There are $p=100$ nodes. In Group 1, data $x_{ij}$ at node $j$ for subject $i$ is simulated as independent standard normal $N(0,1)$ for the both simulations.\\
Study1 (no group difference): In Group 2, we simulated data $y_{ij}$ at node $j$ for subject $i$ as $y_{ij} = x_{ij} + N(0, 0.05^2)$. Tiny noise $N(0, 0.05^2)$ is added to perturb Group 1 data a little bit. It is expected there is no group difference. Following the proposed framework, we constructed the sparse correlation networks and constructed a Betti-plot. Then performed the Jackknife resampling and constructed 20 Betti-plots in each group. The rank sum test was applied and obtained the $p$-value of 0.78. So we could not detect any group difference as expected.\\
Study 2 (group difference):
We first simulate data as $y_{ij} =x_{ij} + N(0, 0.05^2)$ independently for all the nodes. Then for four nodes indexed by  $i=2, 3, 4, 5$, we introduce additional dependency: $y_{ij}=0.5 x_{i1} + N(0, 0.05^2).$
We added small noise to perturb the node values further. This dependency gives any connection between 1 to 5 to have high correlation. Figure \ref{fig:jackknife} shows the simulated correlation matrix. Following the proposed framework, we constructed the sparse correlation networks and constructed a Betti-plot. Then performed the Jackknife resampling and and constructed 20 Betti-plots in each group. The rank sum test was applied and obtained the $p$-value less than 0.001. This significance corresponds to the horizontal gap between the Betti-plots after the filtration value 0.7 (Figure \ref{fig:jackknife} right).

\section{Application}

\subsection{Imaging Data Set and Preprocessing}
The study consists of 23 children who experienced documented maltreatment early in their lives, and 31 age-matched normal control (NC) subjects. 
All the children were recruited and screened at the University of Wisconsin. The maltreated children were raised in institutional settings, where the quality of care has been documented as falling below the standard necessary for healthy human development. 
For the controls, we selected children without a history of maltreatment from families with similar current socioeconomic statuses.  The exclusion criteria include, among many others, abnormal IQ ($< 78$), congenital abnormalities (e.g., Down syndrome or cerebral palsy) and fetal alcohol syndrome (FAS). 
The average age for maltreated children was 11.26 $\pm$ 1.71 years while that of controls was 11.58 $\pm$ 1.61 years. This particular age range is selected since this development period is characterized by major regressive and progressive brain changes \cite{lenroot.2006,hanson.2013}. There are 10 boys and 13 girls in the maltreated group and 18 boys and 13 girls in the control group. Groups did not differ on age, pubertal stage, sex, or socio-economic status \cite{hanson.2013}. 
The average amount of time spent in institutional care by children was 2.5 years $\pm$ 1.4 years, with a range from 3 months to 5.4 years. 
Children were on average 3.2 years $\pm$ 1.9 months when they adopted, with a range of 3 months to 7.7 years. Table \ref{table1} summarizes the participant characteristics.

T1-weighted MRI  were collected using a 3T General Electric SIGNA scanner (Waukesha, WI), with a quadrature birdcage head coil. DTI were also collected in the same scanner using a cardiac-gated, diffusion-weighted, spin-echo, single-shot, EPI pulse sequence. The details on image acquisition parameters are given in \cite{hanson.2013}. Diffusion tensor encoding was achieved using twelve optimum non-collinear encoding directions with a diffusion weighting of 1114 s/mm$^2$ and a non-DW T2-weighted
reference image. Other imaging parameters were TE = 78.2 ms, 3 averages (NEX: magnitude
averaging), and an image acquisition matrix of 120 $\times$ 120 over a field of view of 240 $\times$ 240 mm$^2$. 
To minimize field inhomogeneity and image artifacts, high order shimming and fieldmap images were collected using a pair of non-EPI gradient echo images at two echo times: TE1 = 8 ms and TE2 = 11 ms. For MRI, a study specific template was constructed using the diffeomorphic shape and intensity averaging technique through Advanced Normalization Tools (ANTS) \cite{avants.2008}. Image normalization of each individual image to the template was done using symmetric normalization with cross-correlation as the similarity metric. 
Then the Jacobian determinants of the inverse deformations from the template to individual subjects were computed at each voxel. The Jacobian determinants measure  the amount of voxel-wise change from the template to the individual subjects. White matter was also segmented into tissue probability maps using template-based priors, and registered to the template \cite{bonner.2012}. For DTI, images were corrected for eddy current related distortion and head motion via FSL software (http://www.fmrib.ox.ac.uk/fsl) and distortions from field inhomogeneities were corrected using custom software based on the method given in \cite{jezzard.1999} before performing a non-linear
tensor estimation using CAMINO \cite{cook.2006}. Subsequently, we have used iterative tensor  image registration strategy given in \cite{zhang.2007} and \cite{joshi.2004} for spatial normalization. 
Then Fractional anisotropy (FA) were calculated for diffusion tensor volumes diffeomorphically registered to the study specific template.

The proposed methods have been applied to MRI and DTI of stress-exposed children in characterizing the white matter structural differences against the normal controls.

\begin{table}
\begin{center}
\begin{tabular}{ccc}
\hline
&  Maltreated & Normal controls\\
\hline
Sample size & 23 & 31\\
Sex (males) & 10 & 18 \\
Age (years) & 11.26 $\pm$ 1.71 & 11.58 $\pm$ 1.61 \\
Duration (years) &  2.5 $\pm$ 1.4 (0.25 to 5.4)& \\
Time of adoption (years) & 3.2 $\pm$ 1.9 (0.25 to 7.7)&\\
\hline
\end{tabular}
\end{center}
\caption{Study participant characteristics}
\label{table1}
\end{table}

\begin{figure}[t]
\centering
\includegraphics[width=1\linewidth]{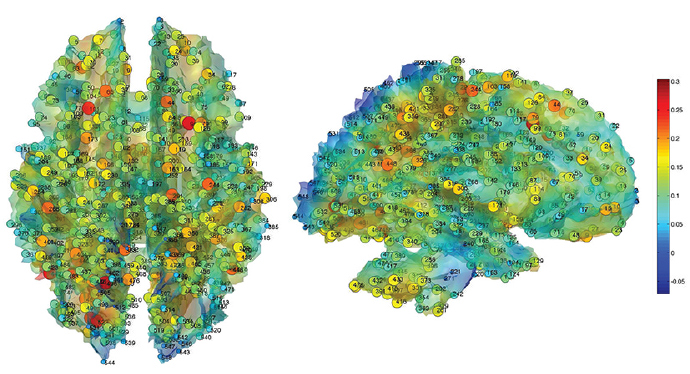}
\caption{548 uniformly sampled nodes along the white matter surface where the sparse correlations and covariances are computed. The nodes are sparsely sampled on the template surface to guarantee there is no spurious high correlation due to proximity between nodes. 
Color scales are the Jacobian determinant of a subject. The same nodes are taken in both MRI and DTI to check the consistency  between the results.}
\label{fig:tstat}
\end{figure}

\subsection{Results: Proposed Sparse Correlation}

We threshold the white matter density at 0.7 and obtained the isosurface. The resulting isosurface is not the gray and white matter tissue boundary and it is located inside the white matter. We are interested in the white matter changes along the tissue boundary. The surface mesh has 189536 mesh vertices and the average inter-nodal distance of 0.98mm. Since Jacobian determinant and FA values  at neighboring voxels are highly correlated, 0.3$\%$ of the total mesh vertices are uniformly sampled to produce $p=548$ nodes.  This gives average inter-nodal distance of 15.7mm, which is large enough to avoid spurious high correlation between two adjacent nodes (Figure \ref{fig:tstat}). The isosurface of the white matter template was extracted using the marching cube algorithm \cite{lorensen.1987}. The number of nodes are still larger than most region of interest (ROI) approaches in MRI and DTI, which usually have around 100 regions \cite{zalesky.2010}.  
This resulted in $548 \times 548$ sample covariances and correlation matrices, which are not full rank. 
We constructed the sparse correlation based network filtrations from the soft-thresholding method using Theorem 3 (Figure \ref{fig:correlation-black}). Subsequently,  Theorem 4 is used to plot the corresponding Betti-plots (Figure \ref{fig:VBM-stroke}).
Since each group produces one Betti-plot, the leave-one-out Jackknife resampling technique was performed to produce 23 and 31 resampled Betti-plots respectively for the two groups. We then computed the areas under the Betti-plots. Using the rank-sum test, we detected the statistical significance of  the area differences between the groups ($p$-value $<$ 0.001). The Betti-plots for normal controls show much higher Betti numbers at any given threshold. 

\begin{figure*}[t]
\centering
\includegraphics[width=0.9\linewidth]{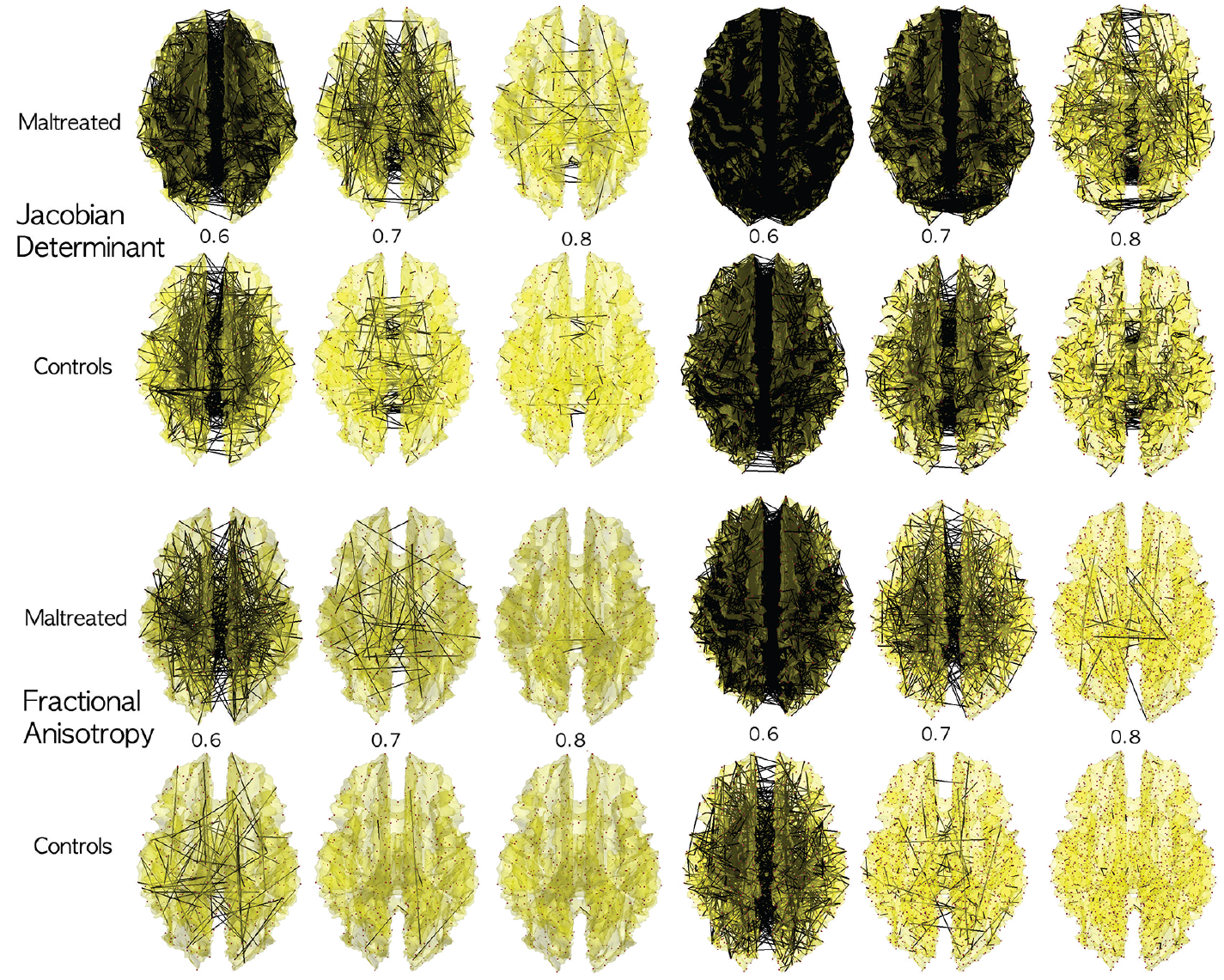}
\caption{Networks $\mathcal{G}(\lambda)$ obtained by thresholding  sparse correlations for the Jacobian determinant from MRI and fractional anisotropy (FA) from DTI at different $\lambda$ values ($\lambda = 0.6, 0.7, 0.8$) for 548 nodes (left three columns) 1856 nodes (right three columns). The collection of the thresholded graphs forms a filtration. The children exposed to early life stress and maltreatment show more dense network at the given $\lambda$ value. Since the maltreated children are more homogenous in the white matter region, there are more dense high correlations between nodes. The over all pattern of dense connections in the maltreated children is similar between the networks of different node sizes and across the different imaging modalities.}
\label{fig:correlation-black}
\end{figure*}

{\em Biological Interpretation.} 
In the Betti-plots (Figures \ref{fig:VBM-stroke}), we obtain more disconnected components for controls than for children in the early stress group for any specific $\lambda$ value. It can only happen if Jacobian determinants show higher correlations in the maltreated children across the white matter compared to the controls.  So when thresholded at a specific correlation value, more edges are preserved in the maltreated children resulting in decreased number of disconnected components. Thus, the children exposed to early life stress and maltreatment show more dense network at a given $\lambda$ value. This is clearly demonstrated visually in Figure \ref{fig:correlation-black}.  If the variations of Jacobian determents are similar across voxels, we would obtain higher correlations. This suggests more anatomical homogeneity across whole brain white matter in the maltreated children. Our finding is consistent with the previous study on neglected children that shows disrupted white matter organization, which results in more diffuse connections between brain regions \cite{hanson.2013}. Lower white matter directional organization in white matter may correspond to the increased homogeneity of Jacobian determinants and FA-values across the brain regions. Similar experiences may cause some areas to be connected to other regions of the brain at a higher degree inducing higher homogeneity in the regions. This type of relational  interpretation can be obtained from the traditional univariate TBM at each voxel. 

\begin{figure*}[t]
\centering
\includegraphics[width=0.85\linewidth]{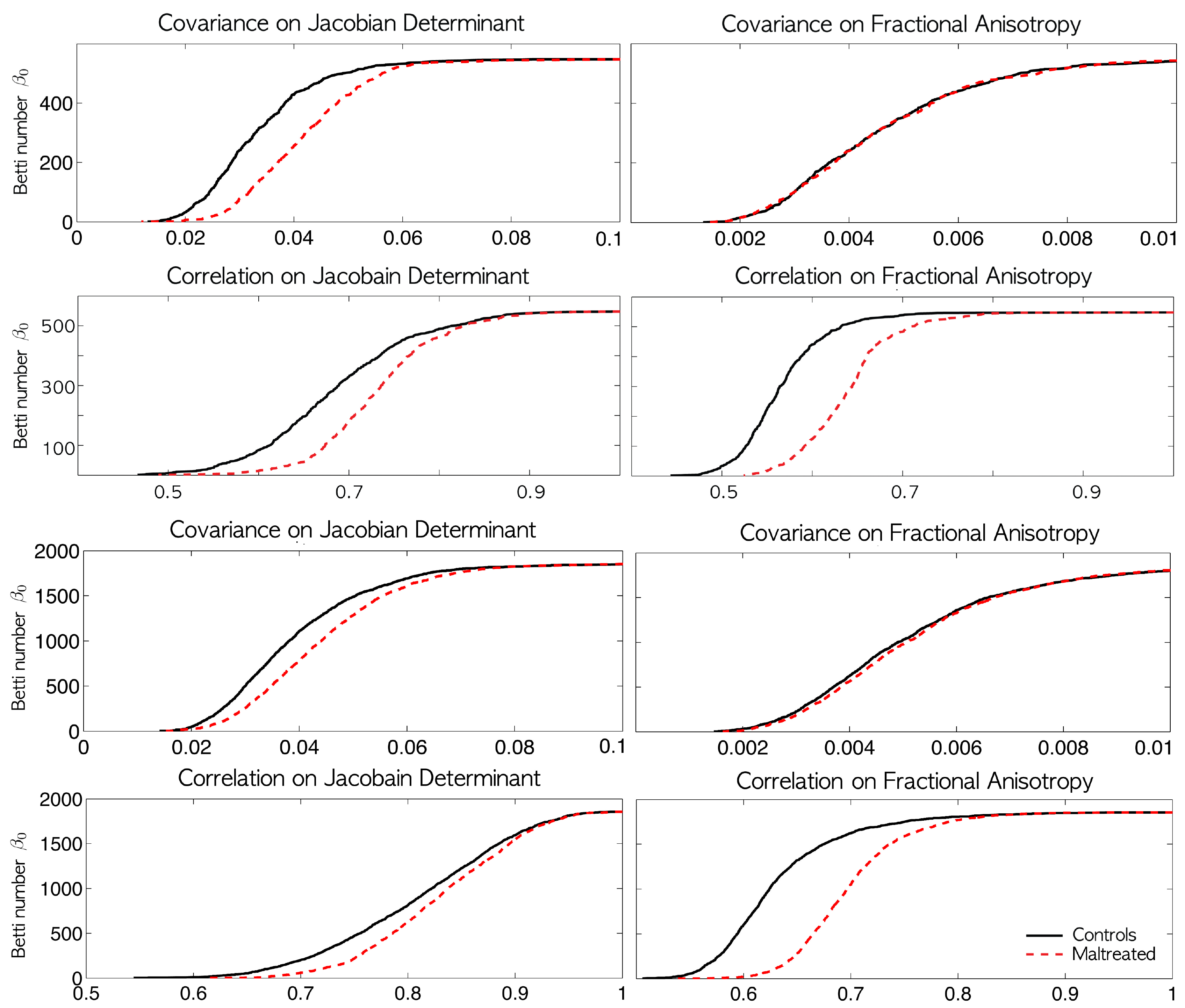}
\caption{The Betti-plots on the sparse covariance and the proposed sparse correlation for Jacobian determinant (left column) and FA (right column) on 548 (top two rows) and 1856 (bottom two rows) node studies. Unlike the sparse covariance, the sparse correlation seems to shows huge group separation between normal and stress-exposed children visually. However, in all 7 cases except top right (548 nodes covariance for FA), we detected statistically significant differences using the rank-sum test on the areas under the Betti-plots ($p$-value $< 0.001$).  The shapes of Betti-plots are consistent between the studies with different node sizes indicating the robustness of the proposed method over changing number of nodes.}
\label{fig:VBM-stroke}
\end{figure*}

\subsection{Comparison Against Sparse Covariance}
We compared the performance of the proposed sparse correlation method to the existing sparse (inverse) covariance method via the penalized log-likelihood \cite{banerjee.2006,banerjee.2008,friedman.2008,
huang.2010,mazumder.2012}, where the log-likelihood  is  regularized with a L1-norm penalty:
\bqn L({\bf \Sigma}) = \log \det {\bf \Sigma} -  \mbox{ tr} \Big(  {\bf \Sigma} S\Big) - \lambda \| {\bf \Sigma} \|_1 \label{eq:banerjee}.\eqn
${\bf \Sigma} = (\sigma_{ij})$ is the covariance matrix and $S$ is the sample covariance matrix. 
$ \| \cdot \|_1$ is the sum of the absolute values of the elements. 
The penalized log-likelihood is maximized over the space of all possible symmetric positive definite matrices. (\ref{eq:banerjee}) is a convex problem and it is numerically optimized using the graphical-LASSO (GLASSO) algorithm \cite{banerjee.2006,banerjee.2008,friedman.2008,huang.2010}. The tuning parameter $\lambda > 0$ controls the sparsity of the off-diagonal elements of the covariance matrix. By increasing $\lambda>0$, the estimated covariance matrix becomes more sparse.

We also  performed the graph filtration technique to the estimated sparse covariance matrix $\widehat{\bf \Sigma} = (\widehat{\sigma}_{ij})$. Let $A= (a_{ij})$ be the adjacency matrix defined from the estimated sparse covariance:
\bqn a_{ij}(\lambda) = 
\begin{cases}
1 &\; \mbox{  if  } \widehat{\sigma}_{ij} \neq 0;\\
0 & \; \mbox{ otherwise.}
\end{cases} \label{eq:Aadj}
\eqn
The adjacency matrix $A$ induces graph $\mathcal{G}(\lambda)$ consisting of $\kappa(\lambda)$ number of partitioned subgraphs: 
\bqn \mathcal{G}(\lambda) = \bigcup_{l=1}^{\kappa(\lambda)} 
G_l (\lambda) \; \mbox{   with } \; G_l =\{ V_l(\lambda), E_l(\lambda) \}, \label{eq:graphG}\eqn
where $V_l$ and $E_l$ are vertex and edge sets of the subgraph $G_l$ respectively. Unlike the sparse correlation case, we do not have full persistency on the induced graph $\mathcal{G}$. The partitioned graphs can be proven to be partially nested in a sense that only the partitioned node sets are persistent  \cite{chung.2013.MICCAI,huang.2010,mazumder.2012}, i.e.
\bqn
V_l(\lambda_1) \supset V_l(\lambda_2) \supset V_l(\lambda_3)  \supset \cdots \label{eq:Vnested}
\eqn
for $\lambda_1 < \lambda_2 < \lambda_3 < \cdots$ and all $l$. 
Subsequently the collection of partitioned vertex set $\mathcal{V}(\lambda )  = \bigcup_{l=1}^{\kappa(\lambda)} V_l (\lambda)$ is also persistent. On the other hand, edge sets $E_l$ may not be persistent. It is unclear if there exists a unique maximal filtration on the vertex set.

The maximal filtration can be obtained as follows. Let  $B(\lambda) = (b_{ij})$ be another adjacency matrix given by
\bqn b_{ij}(\lambda) = 
\begin{cases}
1 &\; \mbox{  if } | \widehat{s}_{ij} | > \lambda;\\
0 & \; \mbox{ otherwise.}
\end{cases} \label{eq:Badj},\eqn
where $\widehat{s}_{ij}$ is the sample covariance matrix. It can be shown that the adjacency matrix $B$ similarly induces graph $\mathcal{H}$  \cite{chung.2013.MICCAI,mazumder.2012}: 
\bqn \mathcal{H}(\lambda) = \bigcup_{l=1}^{\kappa(\lambda)} 
H_l (\lambda) \; \mbox{   with } \; H_l =\{ V_l(\lambda), F_l(\lambda) \} \label{eq:graphH}\eqn
for some edge set $F_l(\lambda)$. 
The subgraphs $G_l$ and $H_l$ have identical vertex set but different edge sets. Then from Theorem \ref{theorem:maximal}, we have maximal filtration on the graph $\mathcal{H}$, where the edge weights are given by the sample covariances. 
Theorem \ref{theorem:maximal} requires the edge weights to be all unique, which is satisfied for the study data set. 
Then similar to Theorem \ref{theorem:PHG}, the Betti-plots are determined by ordering the entries of the sample covariance matrices. The resulting barcode is displayed in Figure \ref{fig:VBM-stroke}. The sparse covariance was also able to discriminate the groups statistically ($p$-value $< 0.001$). The changes in the first Betti number are occurring in a really narrow window but was still able to detect the group difference using the areas under the Betti number plots (Figure \ref{fig:VBM-stroke}). However, the sparse correlations exhibit slower changes in the Betti number over the wide window, making it easier to discriminate the groups.

\subsection{Comparison Against Fractional Anisotropy in DTI}

For children who suffered early stress, white matter microstructures have been reported as more diffusely organized \cite{hanson.2013}. Therefore we predicted less white matter variability in both the Jacobian determinants and also in fractional anisotropy (FA) values as well. The DTI acquisitions were done in the same 3T GE SIGNA scanner; acquisition parameters can be found in  \cite{hanson.2013}. We applied the proposed persistent homological method in obtaining the filtrations for sparse correlations and covariances in the same 548 nodes on FA values (Figure \ref{fig:tstat}). The resulting filtration patterns show similar patterns of a rapid increase in disconnected components for sparse correlations (Figure \ref{fig:correlation-black} and \ref{fig:VBM-stroke}).  
The Jackknife resampling followed by the rank-sum test on the area differences shows  a significant group difference for sparse correlations ($p$-value $<$ 0.001). These results are due to a consistent abnormality among the stress-exposed children that is observed in both MRI and DTI modalities. The stress-exposed children exhibited stronger white matter homogeneity and less spatial variability compared to normal controls in both MRI and DTI measurements. However, the covariance results fail to discriminate the groups at $0.01$ level ($p$-value $=$ 0.043) indicative of a poor performance compared to the sparse correlation method.

\subsection{Robustness on Node Size Changes}

Depending on the number of nodes, the parameters of graph vary considerably up to 95$\%$ and the resulting statistical results will change substantially \cite{fornito.2010, gong.2009,zalesky.2010}. On the other hand,  the proposed method is very robust under the change of node size. For the node sizes between 548 and 1856 (0.3$\%$ and 1$\%$ of original 189536 mesh vertices), the choice of node size did not affect the pattern of graph filtrations, the shape of Betti-plots, or the subsequent statistical results significantly. For example, the graph filtration on 1856 nodes shows a similar pattern of dense connections for the maltreated children (Figure \ref{fig:correlation-black}). The resulting Betti-plots also show similar pattern of the group separation (Figure \ref{fig:VBM-stroke}). The statistical results are also somewhat consistent. For both the Jacobian determinant and FA values, the group differences in Betti-plots obtained from sparse correlations and covariances are all statistically significant ($p$-value $< 0.001$) in both 548 and 1856 nodes except one case. For the case of the 548 nodes covariance on FA values, we did not detect any group differences at 0.01 level ($p$-value = 0.043). On the other hand, we detected the group difference for the 1856 nodes case at 0.001 level. The proposed framework is very sensitive, so it can detect really narrow but very consistent Betti-plot differences.

\subsection{Effect of Image Registration}
We checked how much impact image registration has on the robustness of the proposed method. Anatomical measurements across neighboring voxels are highly correlated within white matter so we do not expect image misalignment will have much effect on the final results.
To determine the variability associated with the image registration, the displacement vector fields from the template to individual brains were randomly perturbed by adding Gaussian noise $N(0, 1)$ to each component. This is sufficiently large noise and causes up to 4mm misalignment for some nodes. Then following the proposed pipeline, the Jacobian determinants are correlated across 548 nodes and Betti-plots are computed. Figure \ref{fig:perturbation} shows five perturbation results. The thick line is without any perturbation. The perturbed Betti-plots are very stable and close to the Betti-plots without any perturbation (thick lines). Th height differences in the perturbed Betti-plots are less than $4.4\%$ in average, which is negligible in the subsequent analysis. In fact, the resulting $p$-values are similar to each other and all the perturbed results detected the group difference ($p$-value $<$ 0.001). Thus, we conclude that the proposed topological framework is robust under image misalignment.

\begin{figure*}[t]
\centering
\includegraphics[width=1\linewidth]{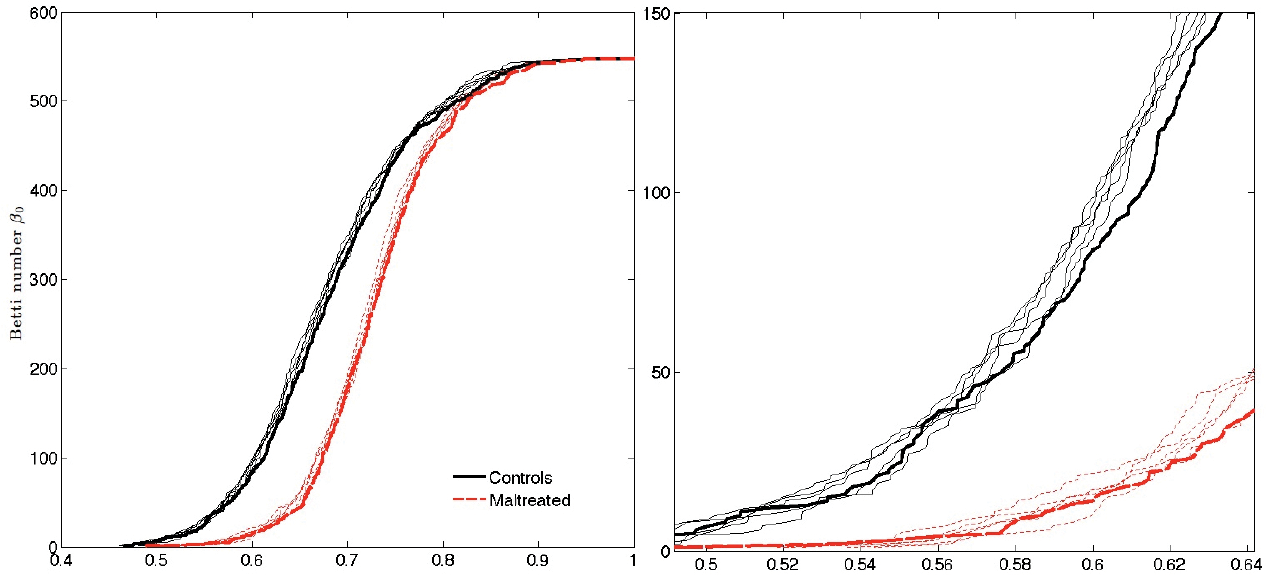}
\caption{The displacement vector field from the template to individual brain is randomly perturbed. Then the Jacobian determinants are correlated across 548 nodes and Betti-plots are subsequently produced. The process is repeated five times to produce five perturbed Betti-plots. The thick line is without any perturbation. The perturbed Betti-plots are very stable and close to the Betti-plots without any perturbation (thick lines). The proposed topological framework is very robust under sufficiently large image misalignment. Right figure is the enlargement of the left figure.}
\label{fig:perturbation}
\end{figure*}

\section{Conclusions and Discussions}

By identifying persistent homological structures in sparse Pearson correlation, we were able to exploit them for speeding up computations.  A procedure that usually takes 56 hours was completed in few seconds without utilizing additional computational  resources. Although we have only shown how to identify persistent homology in the sparse Pearson correlation, the underlying principle can be directly applicable to other sparse models and image filtering techniques. These include the least angle regression (LARS) implementation in more general LASSO \cite{carroll.2009}, heat kernel smoothing  \cite{chung.2010.ni}, and diffusion wavelets  \cite{kim.2012.NIPS}, which all guarantee the nested subset structure over the sparse parameters and kernel bandwidth. We will leave the identification of persistent homology in other frameworks for future studies.

We found that Betti-plots on correlations can visually discriminate better than Betti-plots on covariances. 
In Figure \ref{fig:VBM-stroke}, almost all topological changes associated with the covariance occur in really small range between 0 and 0.1. However, unlike covariances, correlations are normalized by the variances so the topological changes are more spread out between 0 and 1. This has the effect of making the Betti-plots shape differences spread out more uniformly and wide in the unit interval. This is most clearly demonstrated in the covariance {\em vs.} correlation on FA (second column). The Betti-plots of covariances are difficult to discriminate visually because the Betti-plots are squeezed into small range between 0 and 0.1 but the Betti-plots of correlations are more discriminative since the Betti-plots are more spread out. The visual discriminative power comes from the normalization associated with the Pearson correlation. The change in the metric affects the filtration process itself since it is based on the sorted edge weights. Subsequently, the shape of Betti-plots and the statistical inference results also change.

While massive univariate approaches can detect signal locally at each voxel, the proposed graph approach can detect signal globally over the whole brain region. Even though the information obtained by the two methods are complementary, they are somewhat exclusive. The proposed approach tabulates the changes of the number of connected components in the thresholded networks via Betti-plots, which cannot be done at individual node level. Therefore, there is no easy straightforward way of combining or comparing the results from the two methods.
The Betti-plots is a global index that is defined over a whole graph so it cannot be directly applicable  to node-level analysis. However, just like any global graph theoretic indices such as small-worldness and modularity \cite{rubinov.2009.ni, bullmore.2009}, it can be applied to subgraphs around a given node. Thus, it might be possible to measure logical topological characteristic around the node. This is the beyond the scope of the paper and we left it as future research.

This paper is not concerned with white matter anatomical connectivity. Here, we focus on a different issue, namely the degree of interregional dependency of image measurements such as Jacobian determinant and fractional anisotropy across brain regions. The proposed method is general enough to run on any type of volumetric imaging data that is spatially normalized. As an application of the proposed method, we were able to demonstrate developmental differences in brain development among stress-exposed children, who are at known risk for cognitive delays. Our Jacobian determinant results are consistent with DTI. 

There are recent fMRI studies showing head motion to introduce systematic biases in functional connectivity \cite{power.2012,satterthwaite.2012,van.2012}. Motion makes it appear as if long-range connections are weaker than they really are, and short-range connections are stronger than they really are \cite{deen.2012}. However, unlike fMRI, structural image volumes are acquired across such a long time frame, we do not expect the head motion to introduce spurious correlations. Further, we are also not aware of any study that establishes a relationship between head movement and maltreated children. We do not consider the head motion is a concern for our anatomical studies.

\section*{Acknowledgements}
\label{Acknowledgements}
This work was supported by NIH grants MH61285, MH68858, 
MH84051, 
NIH Fellowship DA028087, 
NIH-NCATS grant UL1TR000427 and the Vilas Associate Award. 
The authors like to thank Hyekyoung Lee of Seoul National University and  Matthew Arnold of University of Bristol for the valuable discussions on persistent homology and sparse regressions, and Seung-Goo Kim of Max Planck Institute and Nagesh Adluru of University of Wisconsin-Madison for help with image preprocessing. 
\bibliographystyle{plain}
\bibliography{reference.2015.02.22}

\end{document}